\documentclass[english]{article}
\usepackage{lmodern}

\usepackage[T1]{fontenc}
\usepackage[latin9]{inputenc}
\usepackage{color}
\usepackage{babel}
\usepackage{array}
\usepackage{verbatim}
\usepackage{textcomp}
\usepackage{multirow}
\usepackage{amsthm}
\usepackage{amsmath}
\usepackage{amssymb}
\usepackage{setspace}
\usepackage[unicode=true,
 bookmarks=true,bookmarksnumbered=false,bookmarksopen=true,bookmarksopenlevel=1,
 breaklinks=true,pdfborder={0 0 1},backref=page,colorlinks=true]
 {hyperref}
\hypersetup{pdftitle={On Understanding Generics},
 pdfauthor={Moez A. Abdel-Gawad}}

\makeatletter

\providecommand{\tabularnewline}{\\}

\numberwithin{equation}{section}
\numberwithin{figure}{section}
\newcommand{\lyxaddress}[1]{
\par {\raggedright #1
\vspace{1.4em}
\noindent\par}
}
\newenvironment{lyxcode}
{\par\begin{list}{}{
\setlength{\rightmargin}{\leftmargin}
\setlength{\listparindent}{0pt}
\raggedright
\setlength{\itemsep}{0pt}
\setlength{\parsep}{0pt}
\normalfont\ttfamily}%
 \item[]}
{\end{list}}
  \theoremstyle{remark}
  \ifx\thechapter\undefined
    \newtheorem{claim}{\protect\claimname}
  \else
    \newtheorem{claim}{\protect\claimname}[chapter]
  \fi
  \theoremstyle{plain}
  \ifx\thechapter\undefined
    \newtheorem{thm}{\protect\theoremname}
  \else
    \newtheorem{thm}{\protect\theoremname}[chapter]
  \fi
 \ifx\proof\undefined\
   \newenvironment{proof}[1][\proofname]{\par
     \normalfont\topsep6\p@\@plus6\p@\relax
     \trivlist
     \itemindent\parindent
     \item[\hskip\labelsep
           \scshape
       #1]\ignorespaces
   }{%
     \endtrivlist\@endpefalse
   }
   \providecommand{\proofname}{Proof}
 \fi
\newcommand{\code}[1]{\texttt{#1}}

\makeatother

\providecommand{\claimname}{Claim}
\providecommand{\theoremname}{Theorem}

\begin{document}

\global\long\def\dom#1{\mathcal{#1}}
\global\long\def\lang#1{\mathsf{#1}}

\global\long\def\COOP{\dom{COOP}}
\global\long\def\NOOP{\dom{NOOP}}

\global\long\def\gdom#1{G\dom{#1}}

\global\long\def\GNOOP{\gdom{NOOP}}
\global\long\def\GNOOL{\lang{GNOOL}}

\global\long\def\ext{\blacktriangleleft}
\global\long\def\subsign{\trianglelefteq}

\global\long\def\edom#1{E\gdom{#1}}

\global\long\def\EGNOOP{\edom{NOOP}}

\title{\label{cha:GNOOP}On Understanding Generics: Towards a Simple and
Accurate Domain-Theoretic Model of Generic Nominally-Typed OOP}

\author{Moez A. AbdelGawad\\
\code{moez@cs.rice.edu}}

\maketitle

\lyxaddress{\begin{center}
College of Mathematics and Econometrics, Hunan University\\
Changsha 410082, Hunan, P.R. China\\
Informatics Research Institute, SRTA-City\\
New Borg ElArab, Alexandria, Egypt\\
(Began in 2010, while graduate student at Rice University)
\par\end{center}}

\begin{singlespace}
\begin{flushright}
\emph{\scriptsize{}The Gamma function is a metaphorical extension
of the factorial function. One property, its recursion, becomes its
most important feature, and serves as the basis for extending it.
It\textquoteright s a bit like calling an automobile a `horseless
carriage', preserving its essence of carrying and removing the unnecessary
horsefulness, or like calling a railroad }{\scriptsize{}ferrovia}\emph{\scriptsize{}
in Italian or }{\scriptsize{}Eisenbahn}\emph{\scriptsize{} in German,
focusing on the fact that it\textquoteright s still a road, but one
made instead of iron.}\\
\emph{\scriptsize{}\textasciitilde{}Adapted from a quote by Columbia
Physicist Prof. Emanuel Derman}
\par\end{flushright}{\scriptsize \par}
\end{singlespace}

\medskip{}

One important value behind developing and presenting $\NOOP$~\cite{NOOP,NOOPbook,NOOPsumm,InhSubtyNWPT13,Dom4OOP14,DomThSummCOOP14,AbdelGawad2016,AbdelGawad2015,AbdelGawad2015a},
is to provide a more precise foundation on which further OOP research
can be built. $\NOOP$ provides a capacity for better mathematical
reasoning about mainstream OO languages. %

The development of $\NOOP$ was mainly motivated by the lack of a
precise model of OOP despite the prominence and domination of OOP
among mainstream programmers, and the realization that extant models
of OOP are deeply flawed, inadequate, or not enough for explaining
the typing features of mainstream programming languages, and are thus
hindering further development of these languages.\footnote{The development of the process of adding ``closures'' to the Java
programming language is a prime example of how development of mainstream
OOP languages is hindered. Also, research done on Java wildcards has
been unable to present a proof that beyond doubt convinces us of the
type safety of Java with wilcards. Also, Smith and Cartwright, in
\cite{SmithJava08}, present flaws in the current Java type system,
particularly regarding its polymorphic methods type inference algorithm.

Even further, at least as of version 1.6.0\_16, the standard Java
compiler, \code{javac}, has generics-related bugs. For example, the
following code correctly compiles with no errors on the Eclipse Java
compiler. The same code, however, generates a `\code{Compile exception: java.lang.NullPointerException}'
error message when compiled using \code{javac}
\begin{lyxcode}
class~A<T>~\{\}

class~B<T>~extends~A<C<T>\textcompwordmark{}>~\{\}~//~note~that~C<T>~is~used~inside~the~supertype~of~B<T>

class~C<T>~extends~B<T>~\{\}~~~~//~note~that~C<T>~extends~B<T>
\end{lyxcode}
(See Oracle Java compiler bug report~\cite{AbdelGawadGenericsJavacBug10}
for an updated status of this bug in \code{javac}).

We believe all such problems are due to the lack of a precise conceptual
mathematical model of mainstream OOP, particularly one that takes
nominality in consideration. Even though the problems we mention are
Java-specific problems, we believe similar problems (probably some
of which are undiscovered) do exist in other mainstream OOP languages,
given that the type systems of these languages are much similar to
that of Java.} A feature of mainstream OOP where such inadequacy and hindrance
of development are clear is `generics'.

Given that genericity is a useful feature for most OO developers,
generics (\emph{i.e.}, generic classes) are supported in most mainstream
OO languages. Generics offer OO developers with more expressive type
systems.

\section{Generics: A Summary of Developers Perspective}

The dictionary definition of the word \emph{generic} is: `Applicable
to an entire class or group (not tied to particulars, or unconstrained)'.
For example, if we purchase some generic dish soap, soap that has
no brand name on it, we know that we are buying dish soap and expect
it to help us clean our dishes, but we do not know what exact brand
(if any) will be inside the bottle itself. We can treat it as dish
soap, even though we don't really have any idea of its exact contents.

Similarly, in OOP, generic classes, usually called \emph{generics},
provide the programmers to abstract their classes over some types,
and thus to define them ``generically'', independent of the particular
instantiations of them that class users later actually use. Generics
move the decision as to what actual types to be used for some of the
types used inside the class to the \emph{usage-sites} of a class (\emph{i.e.},
are decided by the users of the class) rather than be declared and
fixed at declaration-sites (\emph{i.e.}, decided by class developers).

Generics also offer OO programmers more flexibility, given that different
type parameters of a generic class can be used at different usage
sites, even in the same program. Without generics, such a capability
could only be simulated by a cooperation between class developers
and class users, depending on OO subtyping (which offers so-called
`subtyping polymorphism') and using the so-called `generic idiom'.
Using the generic idiom is not a type-safe alternative to generics,
given that using it involves requiring class users to insert downcasts
by hand. Because they circumvent the type system, programs with downcasts
can be type unsafe (See~\cite{FJ/FGJ}).

\section{Modeling Generics}

Building on how we modeled nominal OOP in $\NOOP$, and given that
generics is a typing feature of nominal OOP that crucially depends
on nominality (as we will see shortly), it is natural to expect that
modeling generics mathematically builds on top of the typing concepts
we developed in $\NOOP$, particularly on top of $\NOOP$ signatures.

In the following, we will see that this expectation is largely true.
Compared to the modeling of nominal OOP in $\NOOP$, the modeling
of generics (or, more accurately, the modeling of `generic nominal
OOP') is about allowing type variables to be used inside signatures
to stand for (abstract over) other signature names, making the plain
names of signatures ``gain some structure'' by allowing them to
be ``applied'' to other signature names, and, finally, to allow
any valid combination of signature names to appear in any place inside
a signature where a ``plain'' (\emph{i.e.}, $\NOOP$) signature
name was allowed before (\emph{e.g.}, field signatures, and method
signatures). Combinations of signature names are formed, syntactically,
by ``applying'' a generic signature name to other combinations of
signature names\footnote{Generic signatures that take no parameters (``zeroary'' signatures)
are treated as being non-generic signatures. When used, they can be
applied only to the empty sequence of signature names.}.

Making signatures be generic (abstracting them over other signature
names, using type variables) is sometimes called ``generification''.
Using generic signatures, we define ground object signatures, which,
similar to $\NOOP$ object signatures, are paired with object records
to construct generic objects (\emph{i.e.}, objects that carry instantiations
of generic signatures as their nominal typing information).

Similar to how we constructed $\NOOP$, we then build a model of generic
nominal OOP, which we call $\GNOOP$. $\GNOOP$ includes a domain
of generic objects which we call $G\dom O$. Domain $G\dom O$ uses
ground signatures in the construction of its objects. Otherwise, $G\dom O$
is very similar to domain $\dom O$ of $\NOOP$. 

In the next two sections, we formally present generic signatures,
and (a rough sketch for) the formal construction of $\GNOOP$. Construction
of $\GNOOP$ goes along the same lines of constructing $\NOOP$.

As we will see in the next two sections, all features of generics,
and their modeling, are centered around the main crucial idea that
names of signatures are not plain\emph{ }but have some structure,
indirectly implying, thus, that $\GNOOP$ objects could roughly be
viewed as ``$\NOOP$ objects with structured signature names'' (rather
than plain names used in $\NOOP$).

\section{\label{sec:Generic-Signatures}Generic Signatures}

Before giving the formal definitions for generic signatures, we first
present an informal view of generic signatures to help motivate the
later formal definitions.

\subsection{Informal View}

A \emph{signature constructor }is a parametrized (``generified'')
$\NOOP$ signature. A signature constructor carries almost the same
information a $\NOOP$ signature does (\emph{i.e.}, a name, supersignatures,
member signatures, ... etc). Like a $\NOOP$ signature, a signature
constructor\emph{ }is a syntactic typing entity. 

Each signature constructor has a \emph{signature constructor name}
that is a plain label. Signature constructor names\emph{ }have a very
similar purpose in $\GNOOP$ to that of signature names in $\NOOP$
signatures. Given its dependency on the use of names (of signature
constructors), generics is also a feature of nominal OO programming
that pivotally depends on nominality.

The main crucial difference between a signature constructor and a
$\NOOP$ signature is that, inside a signature constructor, some generic
signature names inside the signature constructor might be ``missing''
and have instead \emph{type variables} as place holders (\emph{i.e.},
names are abstracted over by the type variables)\footnote{For the sake of familiarity, we keep using the name `type variables'
for these variables that abstract over signature names, although,
strictly speaking, they should be called `ground signature name variables
(See below).}, and, further, structured generic signature names have to be used
inside signature constructors (rather than plain unstructured signature
names).

Following the name of a signature constructor, %
as an extra component in of a signature, all type variables that might
be used inside the signature constructor are declared (in a sequence
of type variables) ahead of the other components of the signature
constructor (no nesting of type variable declarations occurs except
at the top level, akin to a logical `prenex form'). The sequence
of type variables is ordered (hence, a sequence) but allows no repetitions
(\emph{i.e.}, all type variables of a signature constructor are distinct/have
distinct names). Signature constructors are thus binding constructs.
As such, two signature constructors are considered equivalent if they
are equal modulo the consistent renaming of type variables throughout
the two constructors (``alpha-renaming'').

No type variable can be used/referenced inside a signature constructor
without it being declared in the type variables\emph{ }component of
the constructor. Like a function expression in $\lambda$-calculus,
the signature constructor, via its declared type variables, is said
to \emph{abstract over} the ground signature names that type variables,
at usage-sites, can be instantiated to. Due to the fact that type
variables cannot be ``applied'' (to other type names), a more accurate,
and simpler, mathematical view of signature constructors is that they
are signature \emph{schemes}, rather than functions. We adopt this
simpler view. We will see, in Section~\ref{gen-sig-substitution},
how this view (together with type variables being distinct, and being
declared in prenex form) helps us define a very simple notion of \emph{name
substitution} on signatures. All usage and properties of generic signatures
in generic OOP depend on this simple notion of substitution.

Zeroary signature constructors, \emph{i.e.}, ones which take no type
parameters, play the role of providing the base case for almost all
inductive definitions or proofs that involve generic signatures. Given
it takes no type parameters (i.e., has an empty sequence of type variables),
and thus has no missing information, a zeroary signature constructor,
when its (empty) type variables component is dropped, corresponds
directly to a non-generic $\NOOP$ signature. Vice versa, by adding
an empty sequence of type variables to a $\NOOP$ signature it directly
corresponds to a zeroary signature constructor. (In summary, thus,
each zeroary $\GNOOP$ signature constructor corresponds to a $\NOOP$
signature, and each $\NOOP$ signature corresponds to a zeroary $\GNOOP$
signature constructor. There is a one-to-one correspondence between
the two entities). Thus, in $\GNOOP$, as in~\cite{FJ/FGJ}, ordinary
non-generified (\emph{i.e.}, non-generic) signature names are identified
with zero-ary signature constructor names. This is done merely as
a technical convenience.

A \emph{signature constructor environment }is a finite map (a table)
that maps signature constructor names to signature constructors. Like
all functions/mappings, by applying a signature constructor environment
to a signature constructor name, the corresponding signature constructor
(the one with the input name) is obtained.

To construct a \emph{generic signature name}, for use either inside
the supersignatures component, or the member signatures component
of a signature constructor, a signature constructor name\footnote{Whose corresponding signature constructor is then said to be `instantiated',
\emph{i.e.}, to be made an instance of.} is paired with a (possibly empty) sequence of `type variables, or
(nested) generic signatures names'. This sequence is called the sequence
of \emph{type arguments }that are ``passed'' (as parameters) to
the signature constructor name to construct a generic signature name.

Relative to a certain signature constructor environment, a generic
signature name is \emph{well-formed} if all signature constructor
names inside it are paired with sequences of type arguments that are
the same length as the type variables components of the signature
constructors in the environment corresponding to the signature constructor
names\footnote{That is, for a well-formed generic signature name, the number of type
arguments must be equal to the declared number of type parameters.
The declared number is based on the signature constructor information
inside a signature constructor environment.}.

Generic signature names can be represented as labeled abstract syntax
trees. The signature constructor name used to construct the generic
signature name is the label for the root node of the tree. The subtrees
of this node are trees that each represent an element of the type
arguments sequence. Given that type variables cannot be applied, nodes
representing type variables will always be \emph{leaf} nodes in the
abstract syntax tree representation of a generic signature name.

Inside a signature constructor, `a type variable or a generic signature
name' (usually just called a ``type'', or, more accurately, a ``type
name'') can appear as the signature name of a field (in a field signature),
and as the signature name of a method parameter or the signature name
of a method return value (in a method signature). A type name can
also appear as a supersignature name. To prevent the possibility of
having circular subsigning hierarchies, a plain (also called ``naked'')
type variable is not allowed to appear in the supersignatures component
of a signature constructor\footnote{In~\cite{AllenFirst-ClassApproach03}, it is discussed how MixGen
in fact uses naked type variables in the supersignatures component
of a signature to define signatures of ``mixins'' (as OO \emph{components})
on top of first-class OO generics. In the context of mixins, there
are in fact easily enforced rules that prevent circular nominal subtyping
hierarchies. \cite{AllenFirst-ClassApproach03}~presents an algorithm
that MixGen uses to detect, and prevent, such circularity. Since we
aim to model generics without mixins (\emph{i.e.}, to model ``regular
generics'') in this report, dealing with such a use of naked type
variables is unnecessary.}.

In the context of a certain signature constructor environment, well-formed
generic signature names that have no type variables (are ``variable-free'')
are of special status, and they deserve some special attention: They
are constructed only out of signature constructor names, nested inside
each other via the type arguments component of a generic signature
name. A generic signature name that contains no type variables (\emph{i.e.},
has no missing information) is called a \emph{ground signature name}.
Inductively, a ground signature name is, thus, constructed from a
signature constructor name of a signature constructor that takes no
type parameters (a zeroary signature constructor) paired to the empty
sequence (of type names, as its type arguments), or is constructed
from a signature constructor name of a (generic, non-zeroary) signature
constructor paired to a sequence all of whose members are themselves
ground signature names (as the type arguments of the signature constructor
name).

Formally, we will see below that ground signature names are finite
syntactic entities, and are members of an inductively defined set.
As a subset of generic signature names, ground signature names also
can be represented by labeled trees. Due to the absence of type variables,
all leaves of a tree representation of a ground signature name are
labeled by names of zeroary signature constructors. In fact, given
that we noted above that type variables cannot occur elsewhere in
the tree representing a generic signature name except as leaves, all
nodes of the tree representing a ground signature name are labeled
by signature constructors names only.

Relative to a specific signature constructor environment, well-formed
generic signature names that are not ground signature names are called
\emph{non-ground signature names}. They are generic signature names
that do contain type variables. They can only appear inside signature
constructors (in member signatures, or as supersignatures).

Similar to $\NOOP$ object signatures, a \emph{generic object signature
}is a pair of a ground signature name and a signature constructor
environment. Non-ground signature names are cannot be used to form
object signatures.

A signature constructor whose (1) name is substituted by a ground
signature name, that (2) has its type variables component removed,
and (3) whose type variables are replaced (substituted) by ground
signature names consistently throughout the signature constructor
(\emph{i.e.}, throughout its supersignatures, field signatures, and
method signatures components) is called a \emph{ground signature}.
Recalling that signature constructors are schemes, ground signatures
are \emph{instances} of signature constructors. Ground signatures
have no type variables. Every zeroary signature can be trivially made
into a ground signature.

The instantiation of a signature constructor, with ground signature
names as type arguments, defines a ground signature, whose name is
the ground signature name constructed from the signature constructor
name paired with the sequence of type arguments.

Constructing ground signatures via instantiation of signature constructors
is defined using a simple notion of \emph{name substitution}. Name
substitution substitutes (replaces) type variables inside a signature
constructor with ground signature names.

We can easily see, thus, that in the context of a particular signature
constructor environment, a ground signature name $ggsn$ whose first
component, say $nm$, is the name of a signature constructor $sc$
defines (and is also the name of) a ground signature, say $gs$. $gs$
can be obtained from $sc$ and $ggsn$ by substituting the type variables
$\overline{V}$ inside $sc$ by the type arguments of $ggsn$ (the
second component of $ggsn$, say $\overline{TN}$). When we present
formal definitions below, we denote this substitution/instantiation
operation that defines the ground signature $gs$ corresponding to
$ggsn$ by $gs=\{\overline{V}\mapsto\overline{TN}\}sc$, meaning that,
if $sc$ has the same signature constructor name as $ggsn$, and $\overline{TN}$
are the type arguments of $ggsn$, then $gs$ is defined by substituting
the type variables $\overline{V}$ inside $sc$ by the type arguments
$\overline{TN}$ of $ggsn$.

Name substitution plays a very important role in our mathematical
(\emph{i.e.}, domain-theoretic) modeling of nominal OOP generics.

Given the use of ground signature names to define and name ground
signatures, we should note that non-ground signature names do not
name any actual signature names. Instantiation, and name substitution
are only defined for ground signature names. Although, technically
speaking, \emph{generic signatures} could be defined (as a general
notion that embodies signature constructors as well as ground signatures),
they are not needed much in practice\footnote{See Section~\ref{sub:Well-formed-Generic-Signatures} for one such
important (theoretical/meta-linguistic) need, where, for checking
for well-formedness of a signature constructor in a signature constructor
environment, we need to make sure \emph{all} instantiations of signature
constructor satisfy a specific condition related to the corresponding
instantiations of another\emph{ }signature constructor. This condition
involves checking an infinite number of instantiations, which is generally
not possible, unless we can abstract over all these instantiations
(using type variables, which is indeed possible in this case).}. This is why the instantiations of signature constructors that could
define these generic signatures are disallowed.

Because non-ground signature names are names ``with some missing
information'', in instantiation and name substitution it has to always
be guaranteed that whenever used, non-ground signature names (which
may appear inside signature constructors) are used to construct ground
signature names first, which are then the names used to refer to any
ground signatures, and that when constructing ground signature names
(to define ground signatures) non-ground signature names are never
passed as type arguments \textquotedblleft as is\textquotedblright .

In light of the informal definitions above, we can easily see that
generic OOP depends on names (\emph{i.e.}, nominality) even more than
non-generic OOP.

To summarize the above, for generics we have ten name-dependent definitions.
These are:
\begin{enumerate}
\item Signature constructors (which have names), signature constructor names,
and signature constructor environments (mappings from names to constructors),
\item Type variables (which, incidentally, are also mere names, but ones
not related to nominality of OOP),
\item Generic signature names (which are signature constructor names paired
with sequences of type names),
\item Type names (which are type variables or generic signature names. Note
the circular dependency on generic signature names),
\item Name substitution (which instantiates a generic signature name to
a ground signature name, and a signature constructor to a ground signature),
\item Ground signature names, and ground signatures (which result from instantiation/name
substitution), and
\item Generic object signatures (which are ground signature names paired
with signature constructor environments).
\end{enumerate}
Dependency on names (and nominality) allows named entities to be referenced
and used before they are fully-defined (\emph{i.e.}, they allow circularity/mutual-recursiveness).
It should be noted that generic signature entities have even a heavier
dose of circularity than non-generic ones\footnote{Note, for example, that some of the definitions above mutually-recursively
depend on each other. We have signature constructors (which have signature
constructor names) using generic signature names, which, in turn,
recursively use signature constructor names (a counterpart to this
sort circularity exists for non-generic signatures, via the use of
signature names). Also, generic signature names use type names, which,
in turn, are type variables or, recursively, generic signature names
(Due to the absence of type variables, this sort of circularity has
no counterpart in non-generic signatures)}.

Below we formally present all generic OOP definitions we informally
discussed above.

\subsection{Signature Constructors}

Formally, corresponding to signature equations for $\NOOP$, %
for $\GNOOP$ signature constructors we have
\begin{eqnarray}
\mathsf{SC} & = & \mathsf{N}\times\mathsf{X}^{*}\times\mathsf{GN}^{*}\times(\mathsf{L}\times\mathsf{GNX})^{*}\times(\mathsf{L}\times\mathsf{GNX}^{*}\times\mathsf{GNX})^{*}\label{eq:sig-cons}
\end{eqnarray}
where $\mathsf{SC}$ is the set of signature constructors, $\mathsf{X}$
is a non-empty set of type variables\footnote{Or, more accurately, type variable names, or most accurately, ground
signature name variable names (given that ``type'' variables actually
get instantiated only to ground signature names).} as plain names, and where sets $\mathsf{N}$, $\mathsf{T}$, $\mathsf{L}$
and the set constructors $\times$ and $^{*}$ have the same meaning
as in $\NOOP$ signature equations%
.

$\mathsf{X}^{*}$ is the set of (finite) sequences of type variables.
As a component inside a signature constructor, a member of $\mathsf{X}^{*}$
is the sequence of type variables whose members can be used inside
this signature constructor. Ordering of elements in an element (a
sequence) of $\mathsf{X}^{*}$ does matter (type arguments are matched
with type variables based on the \emph{order} of each in their respective
sequences). Repetitions is not allowed allowed in elements of $\mathsf{X}^{*}$.
Similar to all sequences, elements of a sequence $\overline{V}\in\mathsf{X}^{*}$
can be referenced by their indices, \emph{e.g.}, $V_{i}$, where $V_{i}\in\mathsf{X}$.
$\#$ is a function whose value is the length of a sequence (\emph{e.g.},
the expression $\#(\overline{V})$ gives the size of $\overline{V}$).
Given that repetition is not allowed, a sequence can also be viewed
as a function from its members to their indices (natural numbers).
Thus, for example, for $Y\in\mathsf{X}$, $\overline{V}(Y)$ gives
the index of type variable $Y$ in the sequence $\overline{V}$ of
type variables\footnote{Thus, if we have a sequence of signature names, say $\overline{TN}$,
assuming $\overline{V}$ is a sequence of type variables of the same
length as $\overline{TN}$, the expression $\overline{TN}_{\overline{V}(Y)}$
gives the signature name in $\overline{TN}$ corresponding to a type
variable $Y$ in $\overline{V}$. We use this notation below to define
name substitution.}. As a function, $\overline{V}$ is undefined for indices equal to
or larger than its size, nor for variable names that do not exist
in its range.

For generic signature names, we formally have
\begin{equation}
\mathsf{GN}=\mathsf{N}\times\mathsf{GNX}{}^{*}\label{eq:gen-sig-name}
\end{equation}
where, for `a generic signature name or a type variable' (which,
inaccurately, is sometimes also called a \emph{type}, or a \emph{type
name}), we have 
\begin{equation}
\mathsf{GNX}=\mathsf{GN}+\mathsf{X}\label{eq:gen-sig-name-or-type-var}
\end{equation}
(Note the mutual dependency between $\mathsf{GN}$ and $\mathsf{GNX}$,
and that only members of $\mathsf{N}$, not $\mathsf{X}$, can be
paired with members of $\mathsf{GNX}{}^{*}$).

\subsubsection{BNF Rules for Generic Signatures}

Similar to BNF rules for $\NOOP$ signatures, the BNF rules (which
allow us to name components of generic signatures) corresponding to
the definitions above are:

\begin{tabular}{|lllll|}
\hline 
\multicolumn{1}{|l}{\textsf{sc::= }\code{(}\textsf{nm\code{,}$\mathsf{[\overline{X}]}$\code{,}
$\mathsf{[\overline{gnm}]}$\code{,} $\mathsf{[\overline{gfs}]}$\code{,}
$\mathsf{[\overline{gms}]}$\code{)}}} &  &  &  & \multicolumn{1}{l|}{signature constructors }\tabularnewline
\multicolumn{1}{|l}{\textsf{gfs::=} \code{(}\textsf{a}\code{,} \textsf{gnmx}\code{)}} &  &  &  & \multicolumn{1}{l|}{generic field signatures}\tabularnewline
\multicolumn{1}{|l}{\textsf{gms::= \code{\textsf{(}}b\code{,}\code{$\mathsf{[\overline{gnmx}]\rightarrow gnmx}$}\code{\textsf{)}}}} &  &  &  & \multicolumn{1}{l|}{generic method signatures}\tabularnewline
\multicolumn{1}{|l}{\textsf{gnm::= \code{\textsf{(}}nm\code{,}\code{$\mathsf{[\overline{gnmx}]}$}\code{\textsf{)}}}} &  &  &  & \multicolumn{1}{l|}{generic signature names}\tabularnewline
\multirow{2}{*}{\textsf{gnmx::= \code{\textsf{(}}0\code{,}\code{$\mathsf{X}$}\code{\textsf{)}}|\code{\textsf{(}}1\code{,}\code{$\mathsf{gnm}$}\code{\textsf{)}}}} &  &  &  & type variables or generic\tabularnewline
 &  &  &  & signature names\tabularnewline
\hline 
\end{tabular}

For members of \textsf{gnmx} when they are in printed form, the pairing
with 0 and 1 is usually unnecessary, given that type variables and
generic signature names are usually syntactically distinguishable.\footnote{Note that even if sets $\mathsf{X}$ and $\mathsf{N}$ where the same
set (\emph{e.g.}, plain strings), if a non-algebraic notation (\emph{e.g.},
`\code{(x,y)}' for pairing) is used for expressing generic signature
names (as is the case in most mainstream generic OO languages), type
variables (members of $\mathsf{X}$) are \emph{syntactically} distinguishable,
using a context-free grammar, from generic signature names (members
of $\mathsf{GN}$). That is because type variables are not paired
with a sequence of type arguments, since they cannot be ``applied''
and thus do not take type arguments (signature constructors are, thus,
not ``higher-order''), and also because, strictly speaking, signature
constructors must be applied to a sequence of type arguments (even
zeroary signature constructors have to be applied to the empty sequence
of type arguments). Raw types (see Section~\ref{sub:Raw-Types})
and/or making the empty sequence of type arguments optional affects
the easy distinction between the two entities, but the distinction
is still possible using context-dependent information (\emph{e.g.},
when a signature constructor is used in a signature constructor environment,
which provides the names of \emph{all} signature constructors that
could be referenced inside a signature constructor, and by, possibly,
giving either $\mathsf{N}$ or $\mathsf{X}$ a higher priority {[}Java,
for example, gives higher priority to set $\mathsf{X}${]}).} The pairing is, thus, usually elided when a member of \textsf{gnmx}
is spelled out in this thesis.

\subsection{Ground Signature Names}

A proper\emph{ }subset of generic signature names (\emph{i.e.}, a
subset of set $\mathsf{GN}$) that is of special interest (see Section~\ref{sub:Ground-Signatures}
for more details) is one whose members are generic signature names
that have no occurrences of type variables inside them. We call (members
of) this subset \emph{ground signature names}. Ground signature names
could, more simply, be viewed as ``structured signature names'',
as opposed to the \emph{plain} signature names used in constructing
$\NOOP$.

Thus, formally, for the set $\mathsf{GGN}$ of ground signature names,
we have
\begin{equation}
\mathsf{GGN}=\mathsf{N}\times\mathsf{GGN}^{*}\label{eq:ground-sig-name}
\end{equation}
(where the empty sequence of ground signature names provides the base
case for the definition).

Each member of $\mathsf{GGN}$ is a member of $\mathsf{GN}$, but
not necessarily vice versa (Members of the difference set, $\mathsf{GN}\backslash\mathsf{GGN}$,
are the non-ground signature names, which necessarily contain type
variable occurrences).\emph{}

The BNF rule corresponding to the definition of ground signature names
is:

\begin{tabular}{|ccccc|}
\hline 
\textsf{ggnm::= \code{\textsf{(}}nm\code{,}\code{$\mathsf{[\overline{ggnm}]}$}\code{\textsf{)}}} &  &  &  & ground signature names\tabularnewline
\hline 
\end{tabular}

\subsection{\label{sub:Sig-Con-Env-Gen-Obj-Sigs}Signature Constructor Environments
and Generic Object Signatures}

Similar to $\NOOP$ signature environments, signature constructor
environments (SCEs, members of a set $\mathsf{SCE}$, for short) are
finite functions from $\mathsf{N}$ (signature constructor names)
to $\mathsf{SC}$ (signature constructors), \emph{i.e}., SCEs are
particular subsets of $\mathsf{N}\times\mathsf{SC}$\footnote{Incidentally, we might have used $\multimap$ (the finite records
constructor) to function also as a \emph{set} constructor (making
it ignore any ordering and repetitions of members of its input sets
{[}the defining sets{]} and output sets {[}the defined sets{]}). If
so, then we would have $\mathsf{SCE}=\mathsf{N}\multimap\mathsf{SC}$.
We prefer, in this thesis, however, to limit the use of $\multimap$
to the construction of semantic domains only.}.

Within the context of a particular signature constructor environment
$\mathsf{sce}$, a generic signature name (including ground signature
names) $\mathsf{gnm=(nm,}\mathsf{[\overline{gnmx}]})$ is \emph{well-formed}
if and only if (1) $\mathsf{sce}(\mathsf{nm})$ is defined ($\mathsf{nm}$
is in the domain of $\mathsf{sce}$), (2) $\#(tvars(sce(\mathsf{nm})))=\#(\overline{\mathsf{gnmx}})$,
and (3) all nested generic signature names in $\overline{\mathsf{gnmx}}$
are also well-formed in $\mathsf{sce}$, where $tvars$ is a projection
function that extracts the type variables component of a signature
constructor.

Similar to $\NOOP$ object signatures, generic object signatures\emph{
}are pairs of a ground signature name $\mathsf{ggnm}$ and a well-formed
(See Section~\ref{sub:Well-formed-Generic-Signatures}) signature
constructor environment $\mathsf{sce}$ (where $\mathsf{ggnm}$, by
definition, is well-formed in $\mathsf{sce}$). Generic object signatures\emph{
}are, thus, well-formed members of $\mathsf{GGN}\times\mathsf{SCE}$.
Note that ground\emph{ }signature names, not the more general generic
signature names, are used to define generic object signatures.

Similar to a $\NOOP$ signature, the informal ``meaning'' of a generic
object signature is that its first component, the ground signature
name, should be ``interpreted'' in the context of its second component,
the signature constructor environment (In particular, signature constructor
names referenced by the ground signature name are meant to refer to
names of signature constructors in the signature constructor environment).

\subsection{\label{sub:Ground-Signatures}Ground Signatures}

Ground signatures provide a connection (and a ``middle ground'')
between $\NOOP$ signatures and $\GNOOP$ signatures. Ground signatures
are signature constructors that (1) use ground signature names as
their names\footnote{That is, their names are members of $\mathsf{GGN}$, rather than members
of $\mathsf{N}$.}, (2) have no type variables component, and (3) have type variables
occurrences inside supersignatures, and member signatures substituted
with ground signature names. Ground signatures could, equivalently,
be viewed as $\NOOP$ signatures with structured\emph{ }signature
names.

The equivalence of the two views of ground signatures is why they
are considered a ``middle ground'' between (non-generic) $\NOOP$
signatures and (generic) $\GNOOP$ signatures.

Formally, we have $\mathsf{GGS}$, the set of ground signatures, defined
as 
\begin{equation}
\mathsf{GGS}=\mathsf{GGN}\times\mathsf{GGN}^{*}\times(\mathsf{L}\times\mathsf{GGN})^{*}\times(\mathsf{L}\times\mathsf{GGN}^{*}\times\mathsf{GGN})^{*}\label{eq:ground-sigs}
\end{equation}
which is the same\emph{ }as the $\NOOP$ signature equation%
{} but uses ground signature names (members of $\mathsf{GGN}$) in place
of plain signature names (members of $\mathsf{N}$), implying names
of ground signatures have an inherent, desirable structure, and are
not plain names as those of $\NOOP$ signatures.

The fact that ground signatures provide a middle ground (expressed
mathematically in the equations for $\mathsf{S}$, $\mathsf{SC}$
and $\mathsf{GGS}$) is behind the statement that ``generics saves
developers typing {[}as in writing{]} time'' (which has the implication
of decreasing the chances for coding errors, as well), and that ``generics
does not let developers do something fundamentally different than
what they were able to do without\emph{ }generics (but with significantly
more code written)''\footnote{However, as a refactoring technique, generics has the implication
of decreasing the size of the code that OO developers would have to
maintain if they were to ``encode'' generics into a non-generic
OO language, \emph{e.g.}, using the so-called ``generics idiom''.}.\footnote{Given that sets $\mathsf{N}$ and $\mathsf{GGN}$ have the same cardinality,
$\aleph_{0}$, a one-to-one correspondence (mapping) between the two
sets can be used to define a one-to-one correspondence between ground
signatures and (non-generic) signatures (of $\NOOP$), and, accordingly,
to define a mapping from ground object signatures (see below) to (non-generic)
object signatures. An example of functions that map ground signature
names to non-generic signatures names, and vice versa, are the so-called
``name mangling'' functions, such as the one used in the NextGen
implementation of Java Generics (See~\cite{AllenEfficient02}).}

The BNF rule corresponding to the definition above is:

\begin{tabular}{|lllll|}
\hline 
\multicolumn{1}{|l}{\textsf{ggs::= }\code{(}\textsf{ggnm\code{,} $\mathsf{[\overline{ggnm}]}$\code{,}
$\mathsf{[\overline{ggfs}]}$\code{,} $\mathsf{[\overline{ggms}]}$\code{)}}} &  &  &  & \multicolumn{1}{l|}{ground signatures }\tabularnewline
\multicolumn{1}{|l}{\textsf{ggfs::=} \code{(}\textsf{a}\code{,} \textsf{ggnm}\code{)}} &  &  &  & \multicolumn{1}{l|}{ground field signatures}\tabularnewline
\multicolumn{1}{|l}{\textsf{ggms::= \code{\textsf{(}}b\code{,}\code{$\mathsf{[\overline{ggnm}]\rightarrow ggnm}$}\code{\textsf{)}}}} &  &  &  & \multicolumn{1}{l|}{ground method signatures}\tabularnewline
\hline 
\end{tabular}

\subsubsection{Are Ground Signature Environments and Ground Object Signatures Needed?}

Given the definition of ground signatures, as the result of instantiating
the type variables inside a signature constructor, it may be tempting
to extend this connection between generic and non-generic signatures
to consider defining what could be called `ground signature environments',
which could, it may be thought, also provide a ``middle ground''
between signature constructor environments and $\NOOP$ signature
environments, and could also be, it may be further thought, ``$\NOOP$
signature environments with structured signature names''.

While defining more ground signature entities is possible mathematically,
we should note that a hypothetical ground signature environment corresponding
to a particular signature constructor environment could be an \emph{infinite}
function (from ground signature names to ground signatures), unlike
a $\NOOP$ signature environment, which must be a finite function%
.

An example that demonstrates that a signature constructor environment
can get translated to an infinite ground signature environment is
the case where inside the signature constructor environment there
is a signature constructor that has inside it a generic signature
name in which the main signature constructor name is used to construct
type arguments that are then passed again to the main signature constructor
name (thus making rewriting/name substitution \emph{expansive}. See~\cite{GEB}).
Trying to instantiate this signature constructor to produce a ground
signature, and, recursively, have ground signature names inside it
produce ground signatures that should be in the co-domain of the ground
signature environment will cause an infinite number of different instantiations,
and thus the domain and co-domain (or, range) of the ground signature
environment will also be infinite sets.\footnote{A sample valid Java code that demonstrates a case where rewriting,
if done, would be infinite, is the following code:
\begin{lyxcode}
class~C<T>~\{

~~C<C<T>\textcompwordmark{}>~m()\{~return~new~C<C<T>\textcompwordmark{}>();~\}~//~note~the~expansive~return~type

\}
\end{lyxcode}
This example is similar to the \emph{recursive generics }example in~\cite{AllenEfficient02}
(called \emph{recursive polymorphism} in that paper), but this example,
unlike the one in~\cite{AllenEfficient02}, demonstrates the use
of an expansive type in the \emph{return type} of a method, rather
than only inside its body (if used only inside the body of a method,
a recursively polymorphic type may not necessarily affect the typing
(\emph{i.e.}, signature) of the method).}

In any one generic nominal OOP program, up to any (finite) moment
in time, however, only a \emph{finite} number of instantiations of
a signature constructor could take place, and thus the whole set (of
ground signatures) resulting from all possible instantiations is not
needed (or is ``lazily'' needed). This situation is much like an
inductively defined data type (of, say, lists, or natural numbers),
which has a finite set of data constructors that are then used to
construct values of that data type\footnote{In our case, \emph{i.e.}, for the ``data type'' of ground signature
names, names of zeroary signature constructors provide the base elements
of the data type, while the names of non-zeroary signature constructors
will be used (with the more basic elements) to inductively construct
the compound elements of the data type.}. In fact, we could explicitly formalize ground signature names, relative
to a fixed finite subset of $\mathsf{N}$, as being elements of such
a data type. Given its irrelevance to the goal of modeling generics,
however, except for this brief note we do not present any further
discussion of such a possibility in this thesis.

`Ground object signatures' could also be defined, by pairing ground
signature names and (the undefined) ground signature environments.
Given the impracticality of the latter, the former become impractical
too, so we do not define them here.

\subsubsection{The Import of Ground Signatures and Their Names}

The reason behind defining ground signature names is that they are
the names paired with signature constructor environments (SCEs) to
construct generic object signatures. Generic object signatures, in
turn, are paired with object records to construct the domain of (generic)
objects, $\gdom O$, in $\GNOOP$ (See Section~\ref{sec:GNOOP}).

By using a more structured name for object signatures\emph{ }(\emph{i.e.},
ground signature names), when compared to non-generic OOP, generics
could be viewed as basically offering a ``structuring'' of the namespace
of signature names, and them providing access to a (finitely-constructed)
infinite set of these signatures. Signature constructors, via generic
signature names, then, enable abstracting over subsets of these ground
signatures (subsets whose members share common signature constructor
names). Signature constructors could, thus, be considered as functions
from ground signature names to ground signatures, or, given their
syntactic nature, as \emph{signature schemes} rather than functions
as we pointed out earlier. A signature constructor is an ``abbreviation''
for (abstracts over) an infinite set of structurally-similar ground
signatures. By instantiating a signature constructor, we obtain an
instance of the set of ground signatures that the signature constructor
abbreviates (is a \emph{scheme} for).

Thus, by having structured signature names, generics adds expressiveness
to an OOP language  by allowing \emph{multiple} (ground) signatures
with the same signature constructor name to be usable, in a type safe
manner, in a single OO program.

\subsubsection[Name Substitution and Constructing Ground Signatures]{\label{sub:Name-Substitution}Constructing Ground Signatures: Using
Name Substitution to Instantiate Signature Constructors }

This section presents the formal definition of the name substitution
function. As explained above, name substitution uses ground signature
names to map generic signature names to ground signature names, by
replacing type variable occurrences inside the generic signature names
with ground signature names, and, in turn, mapping a signature constructor
to a ground signatures (\emph{i.e.}, to an instance of the set of
ground signatures that the signature constructor abstracts over).

Given two special properties of how type variables are used in generics,
generic name substitution has a simple definition (unlike the definition
of substitution for $\beta$\nobreakdash-reduction in $\lambda$\nobreakdash-calculus,
for example). These two special properties are:
\begin{enumerate}
\item Type variables are declared all at once, at the ``outer level'',
with nesting of type variables not being allowed except at that level
(type variables later in a sequence of type variables are in the scope
of type variables declared earlier in the sequence).
\item Type variables are required to have distinct names, and thus variable
name clashes (and the use of shadowing, or overriding, or so) between
type variables names are not possible (This property, together with
the first property, allows substitution to be defined in a straightforward
manner without concern for ``capturing'' or shadowing).
\end{enumerate}
Formally, we use the following notation to denote the operation of
substituting type variables $\overline{V}$ in a signature constructor
$sc$ with type arguments $\overline{TN}$
\[
\{\overline{V}\mapsto\overline{TN}\}sc
\]
We first define name substitution on generic signature names, then
extend this definition to define name substitution for signature constructors.
For type variables $\overline{V}\in\mathsf{X^{*}}$ , and $\overline{TN}\in\mathsf{GGN}$,
we have
\[
\{\overline{V}\mapsto\overline{TN}\}Y=\begin{cases}
\overline{TN}_{\overline{V}(Y)} & Y\in\overline{V}\\
Y & Y\notin\overline{V}
\end{cases}
\]
(In fact, the second case, $Y\notin\overline{V}$, should not occur.
The substitution in this case should be undefined \footnote{And an error should be reported, meaning, for example, that the name
substitution is not the result of a proper signature constructor instantiation,
or that the signature constructor is not well-formed.}). For a generic signature name $gnmx=(nm,[\overline{gnmx}])$, we
have
\[
\{\overline{V}\mapsto\overline{TN}\}gnmx=(nm,[\overline{\{\overline{V}\mapsto\overline{TN}\}gnmx}])
\]
\label{gen-sig-substitution}This definition extends naturally to
signature constructors as follows:

For a signature constructor $sc=(nm,[\overline{V}],ssn,gfss,gmss])$,
we have
\begin{eqnarray*}
\{\overline{V}\mapsto\overline{TN}\}sc & = & (nm,[\overline{V}],\{\overline{V}\mapsto\overline{TN}\}ssn,\{\overline{V}\mapsto\overline{TN}\}gfss,\{\overline{V}\mapsto\overline{TN}\}gmss)
\end{eqnarray*}
(note that we require the sequence $\overline{V}$ of the substituted
variables to be the same as the type variables of $sc$, and also
note that the substitution can be performed without need for a signature
environment)\\
where for a generic member (field or method) signature $ms=(c,gcs)$,
$\{\overline{V}\mapsto\overline{TN}\}ms=(c,\{\overline{V}\mapsto\overline{TN}\}gcs)$,
and where for a name substitution on sequences of elements $[\overline{e}]$,
of size $n$, we have
\[
\{\overline{V}\mapsto\overline{TN}\}[\overline{e}]=[\{\overline{V}\mapsto\overline{TN}\}e_{0},\cdots,\{\overline{V}\mapsto\overline{TN}\}e_{i},\cdots,\{\overline{V}\mapsto\overline{TN}\}e_{n-1}]
\]
\global\long\def\gtog{GenToGnd}
For instantiating a signature constructor to define a ground signature,
we can thus, finally, now define
\[
GenToGnd(sc,\overline{TN})=drop\_tvars(\{tvars(sc)\mapsto\overline{TN}\}sc)
\]
where $tvars$ returns the third (type variables) component of a signature
constructor, and $drop\_tvars$ is a function that takes a sextuplet
and returns a quintuplet without the third component of the input
(drops it), and, by extension to generic object signatures, for a
generic object signature $gos=(ggnm,sce)$ we define a ground signature
$ggs$ where
\begin{eqnarray*}
\gtog(gos) & = & \gtog(sce(fst(ggnm)),snd(ggnm))
\end{eqnarray*}
where $fst$ and $snd$ return the first and second component of a
pair, respectively (and thus for the ground signature name $ggnm$
they return the signature constructor name used to construct $ggnm$,
and the type arguments to that constructor name, respectively).
\begin{claim}
It is relatively clear that for a generic object signature $gos$,
$\gtog(gos)$ is a ground signature with ground signature name $fst(gos)$.
\end{claim}
For a generic object signature $gos$, one may also try defining a
ground signature environment that, if paired with the ground signature
name of $gos$, would define a ground object signature. This definition
would have to produce infinite ground signature environments for some
generic object signatures, \emph{i.e.}, ones containing signature
constructors with `recursive generics' in their signature constructor
component.

It should be noted that the $\gtog$ function loses some information,
and is thus not a reversible function. More concretely, this means
that while a generic object signature can be mapped to a unique ground
signature, it is not possible to reverse the work of $\gtog$ by mapping
the resulting ground signature back to a unique generic object signature,
and thus it cannot be mapped back to the same generic object signature
that produced it. Intuitively, this is impossible because there are
multiple ways to ``generify'' a non-generic ground signature (\emph{i.e.},
many ways to make a non-generic signature be generic). Because of
such an information loss, $\gtog$ is called an \emph{erasure} function.

\subsection{\label{sub:Well-formed-Generic-Signatures}Well-formed Generic Signatures}

Similar to what we have for $\NOOP$, only \emph{well-formed} generic
signatures are used in constructing $\GNOOP$. Well-formed generic
signatures satisfy seven constraints.

(1-4) The first four well-formedness constraints for $\NOOP$%
{} are almost the same for generic signature constructors and signature
constructor environments, with the only difference being that names
(members of set $\mathsf{N}$) now refer to signature constructors
not to signatures.

(5) We also require, for a signature constructor to be well-formed,
that all generic signature names inside it to be well-formed (See
Section~\ref{sub:Sig-Con-Env-Gen-Obj-Sigs}), and for type variables
inside the signature constructor to be members of its type variables
component.

(6) The generic counterpart of the (No cycles in supersignatures)
condition%
{} requires that when the supersignatures hierarchy of a signature constructor
is recursively followed, via signature constructor names inside the
signature constructors, and via the bindings of these in the signature
environment, each signature constructor in the co-domain of a well-formed
signature constructor environment has to have no cycles in its list
of supersignatures (\emph{i.e.}, it has no `infinite ascending chains'
of ``super'' signature constructors). The names of the supersignature
constructor names are obtained from the first component (the ``head'')
of the generic supersignature names.

(7) The generic counterpart of the (Signatures inherit signatures
of members in supersignatures) condition%
{} is the most interesting well-formedness condition. Like for $\NOOP$,
members in well-formed signature constructors that are inherited from
(have the same name in) supersignature constructors are required to
have the \emph{same} signature they have in the supersignature constructors
when the supersignature constructors are instantiated with the (generic)
signature names declared in the subsignature. This condition can
be checked by making a name substitution (See Section~\ref{sub:Name-Substitution})
on the supersignature constructor using the \emph{generic} signature
name parameters used in the supersignatures component of the subsignature
constructor. The signatures of members of the generic signature resulting
from the substitution will be checked for being exact matches with
signature of the corresponding members in the subsignature constructor.
By using a \emph{generic} signature name (which abstracts over all
possible instantiation of the signature construction) in the substitution,
it is made sure that \emph{all} instantiations of the subsignature
constructor satisfy the ``member signatures matching'' condition
of the corresponding instantiation of the supersignature constructor.
If signatures of all members of a signature constructor match the
signatures of members the constructor shares with all its supersignature
constructors, the signature constructor is well-formed.

\subsection{Equality of Generic Signatures and Extension of SCEs}

Equality of well-formed generic signature entities is defined not
much differently from the definition of equality of $\NOOP$ signatures,
except that the exact names of type variables in signature constructors
are irrelevant, and thus equality is defined modulo ``alpha-renaming''
of type variables. Otherwise, equality is defined component-wise using
equality of components of generic signatures as syntactic elements
(just like for $\NOOP$ signatures).

Given that well-formed signature constructor environments are ``minimal'',
by well-formedness conditions, extensions of signature constructor
environments are defined exactly in a similar manner to their definition
for $\NOOP$ signature environments. A signature constructor environment
$sce_{1}$ extends a signature constructor environment $sce_{2}$
if, modulo alpha-renaming of type variables, all members of $sce_{1}$
(name/signature constructor pairs) are members of $sce_{1}$ (and
thus $sce_{1}$ is a superset of $sce_{2}$).

\subsection{\label{sub:Generic-Inheritance-and-Subsigning}Inheritance and Subsigning
of Generic Signatures}

In a signature constructor environment $\mathsf{sce}$, a signature
constructor $\mathsf{sc}$ has a supersignatures component, $\mathsf{ss}$,
that specifies for each instantiation of $\mathsf{sc}$, which defines
a ground signature $\mathsf{gss}$, what ground signatures are the
ground supersignatures of that instantiation. Ground supersignatures
are decided by the supersignatures component of $\mathsf{gss}$.
When the substitution is, transitively, done throughout the supersignatures
hierarchy (which is of finite size, by well-formedness rules), all
ground supersignatures of $\mathsf{gss}$ can be obtained. In agreement
with nominal subtyping, a pair of ground signatures $\mathsf{(gss_{1}},\mathsf{gss_{2}})$
is in the subsigning relation if and only if the ground signature
name of $\mathsf{gss_{2}}$ is a member of the set of names of the
ground supersignatures of $\mathsf{gss_{1}}$.

The supersignatures component of a signature constructor is thus considered
to define a \emph{pointwise }subsigning relation, between ground signature
instantiations of the signature constructor and instantiations of
its explicitly declared supersignature constructors..

Thus, for two well-formed ground object signatures $gos_{1}=(ggnm_{1},sce_{1})$
and $gos_{2}=(ggnm_{2},sce_{2})$, we have
\begin{equation}
gos_{2}\subsign gos_{1}\Leftrightarrow sce_{2}\ext sce_{1}\wedge(ggnm_{2}=ggnm_{1}\vee ggnm_{1}\in gss(sce_{2},ggnm_{2}))\label{eq:generic-subsigning}
\end{equation}
where $gss$ is a projection function that computes the set of\emph{
}all\emph{ }ground supersignatures names of a given ground signature
name ($ggnm_{2}$) in a given signature constructor environment ($sce_{2}$)
by, transitively, going throughout the supersignature hierarchy. $gss$
uses name substitution to obtain these ground supersignatures names.

\section{\label{sec:GNOOP}$\protect\GNOOP$: A Domain Theoretic Model of
Generic OOP}

Construction of $\GNOOP$ mostly resembles the construction of $\NOOP$.
The main difference between generic objects of $\GNOOP$ and non-generic
objects of $\NOOP$ is the use of generic object signatures in constructing
generic objects in place of the (non-generic) object signatures used
in constructing $\NOOP$ objects.

Given its great similarity to $\NOOP$ construction, we only skim
through the construction of $\GNOOP$ (in the next two subsections).
We then discuss the properties of $\GNOOP$ in Section~\ref{sub:Properties-of-GNOOP}.

\subsection{$\protect\GNOOP$ Domain Equation}

The domain equation used to construct $\GNOOP$ is the same as the
$\NOOP$ domain equation, with a domain $\gdom O$, of generic objects,
replacing domain $\dom O$, and using the (flat) domain of generic
object signatures, $\mathcal{S_{\gdom O}}$, in place of domain $\mathcal{S_{O}}$
of $\NOOP$ object signatures. The domain $\mathcal{L}$ and domain
constructors $\times$, $\multimap$, $\rightarrow$, $^{*}$ retain
their same meanings.
\begin{equation}
\gdom O=\mathcal{S_{\gdom O}}\times(\mathcal{L}\multimap\gdom O)\times(\mathcal{L}\multimap(\gdom O^{*}\rightarrow\gdom O))\label{eq:GNOOP-Domain-Equation}
\end{equation}
Similar to the construction of $\COOP$%
~\cite{NOOP,NOOPbook,DomThSummCOOP14}, a model that represents the
`core' of $\GNOOP$ could be constructed, using the domain equation
\begin{equation}
\gdom O=\mathcal{L}\multimap(\gdom O\rightarrow\gdom O)+\mathcal{B}\label{eq:core-GNOOP}
\end{equation}
but the model constructed using this domain equation would be precisely
the same as $\COOP$, because $\COOP$ uses no signatures, and so
does Equation~\eqref{eq:core-GNOOP}. Thus, like it does for $\NOOP$,
$\COOP$ functions as a simple, structural core of $\GNOOP$ as well.

\subsection{$\protect\GNOOP$ Construction}

The construction method of $\GNOOP$ is the same method used for constructing
$\NOOP$, where the construction proceeds in iterations, guided by
the structure of the right-hand-side of domain equation~\eqref{eq:GNOOP-Domain-Equation}
(which has precisely the same structure as that of the right-hand-side
of the domain equation%
{} of $\NOOP$).

Similar to filtering of pre\nobreakdash-$\NOOP$ to define $\NOOP$%
, $\GNOOP$ filtering is responsible for ensuring generic object records
are matched with concrete generic object signatures where member signatures
in the generic object signatures match those in the object records
(note that generic object signatures define ground signatures, whose
member signatures are used to decide the matching).

\subsection{\label{sub:Properties-of-GNOOP}Properties of $\protect\GNOOP$}

Similar to our investigation of some main properties of $\NOOP$%
, and using terminology used there, we also interpret generic object
signatures of $\GNOOP$ as denoting nominal object types, where we
will have

\begin{equation}
\mathbb{GS}[gos]=\{go\in\gdom O|\exists goss\in\dom{S_{\gdom O}},gr\in\gdom R.(goss\subsign gos)\wedge(go\sqsubseteq(goss,gr))\}\label{eq:gnd-sig-semantics}
\end{equation}
where $\gdom R$ is the generic object records counterpart of domain
$\dom R$ of $\NOOP$ object records.

Similar to what we did for $\NOOP$%
, it is easy, using essentially the proof we developed for $\NOOP$
(only slightly appropriately adjusted for $\GNOOP$), to prove that
generic nominal types denoted by generic object signatures are weak
ideals in domain $\gdom O$.

Similar to what we did for $\NOOP$%
, based on the definition of the interpretation of a $\GNOOP$ generic
object signatures and the definition of the interpretation of a $\NOOP$
non-generic signature being very similar, we also use essentially
the same proof we developed for $\NOOP$ to prove that in $\GNOOP$
the inheritance and generic nominal OO subtyping are completely reconcile
(\emph{i.e.}, that adding generics  to a non-generic nominal OO language
preserves the reconciliation).

Formally, thus, similar to the `inheritance is subtyping' theorem
for $\NOOP$%
, for $\GNOOP$ we also have that
\begin{equation}
gos_{1}\subsign gos_{2}\Leftrightarrow\mathbb{GS}(gos_{1})\subseteq\mathbb{GS}(gos_{2})\label{eq:g-inh=00003Dg-subt}
\end{equation}

That is, one fully instantiated class type is a subtype of another
fully instantiated type only if its fully instantiated signature subsigns
the former (the $\Leftarrow$ direction in~\ref{eq:g-inh=00003Dg-subt}).
If subsigning holds, then the instantiated supertype signature name
must appear in the chain of instantiated supertypes for the subtype.
Then signature matching implies that the weak ideal of the subtype
is contained in the weak ideal of the supertype (the $\Rightarrow$
direction in~\ref{eq:g-inh=00003Dg-subt}).

\section{\label{sec:Modeling-More-Generics}Modeling More Generic OOP Features}

Although the modeling and analysis of more features of generic OOP
is mostly considered future work and beyond the scope of this thesis,
we here present a necessarily-incomplete basis for such future work
and we present some main ideas we intend to build on and develop in
the near future.

First, given that Java implements a somewhat-limited notion of generics
on a virtual machine that has no notion of generics (the JVM~\cite{JVM2}),
using a technique and concept called \emph{erasure}\footnote{Generics, as implemented in Java is sometimes, thus, called ``second-class
generics''.}, we would like to make a full assessment of this technique, by building
a model for \emph{erased generics}. Building on concepts we presented
in earlier sections of this report, we present in Section~\ref{sub:Erasure}
the basis on which we can build such a model of erased generics\emph{
}(which is called $\EGNOOP$). We will see that erased generics (generics
implemented via erasure) provides an approximation\footnote{Not in a domain-theoretic sense!}
for full generics, since it basically `saves programmers from using
the ``generic idiom'' but does not give them the full power of generics',
including support for generic type-dependent runtime operations (such
as creating objects of type variable types {[}\emph{i.e.}, type specified
using naked type variables{]}).

Next, in Section~\ref{sub:Bounded-Type-Variables-F-bounded-Generics},
we will discuss how to incorporate bounding type variables, so that
rather than allowing a signature constructor to be instantiated using
any ground signature name (which has been the case so far), some subtyping
constraints about possible type argument instantiations is made on
type variables, and these constraints are guaranteed to be satisfied
for passed-in type arguments.

`Raw types' is not only an artifact of erasure, but also arise
in Java from the need for supporting ``migration compatibility'',
where erased and non-erased code need to co-exist. We discuss raw
types in Section~\ref{sub:Raw-Types}.

Wildcards, %
{} as a main motivation behind developing $\NOOP$ and $\GNOOP$, try
to increase the expressiveness of generics. Wildcards lessen the
mismatch between subtyping polymorphism and generics, by allowing
unknown types (with, or without known upper and lower bounds) to be
passed as type arguments to signature constructors. Using generics
terminology of Section~\ref{sec:Generic-Signatures}, wildcards allow
the naming and the referencing of more ground signatures. Wildcards
are discussed in Section~\ref{sub:Wildcards}.

In Section~\ref{sec:Polymorphic-Methods}, we briefly discuss how
polymorphic methods could be modeled. Even though not a must, polymorphic
methods are supported in virtually all OO languages that support generics.
Generics allows abstracting full objects (\emph{i.e.}, all their members)
over type arguments, while polymorphic methods allow abstracting individual
methods of objects over type arguments.

While the discussions in the sections on modeling erasure and polymorphic
methods are more technical, the discussions are less technical in
the rest of the sections on modeling other generic OOP features.

We start the discussion of more generic OOP features by discussing
erasure.

\subsection[Erased Generic OOP]{\label{sub:Erasure}Erasure and Erased Generic OOP (The ``Generics
Idiom'')}

Erasure as a technique using in some mainstream OO languages, most
notably Java, for supporting generics when first-class generics is
not possible to support. Erased generics thus only approximates full
generics (when information is lost due to erasure, an operation the
depends on the erased information is disallowed, or a warning is produced).
Erasure produces code that a non-generic OO developer ``would have
written (using the `generic idiom')'' if they had no generics.
This is the essential intuition behind erased generics in Java. The
definitions accurately capture this intuition.

The generics idiom, and erased generics accordingly, crucially depend
on the existence of a type (a signature constructor) at the top of
the subtyping hierarchy (typically \code{Object}).\footnote{This is why Java Generics does not support generics using primitive
types, because Java primitive types of Java are \emph{outside} the
object subtyping hierarchy, and are not subtype of \code{Object}.} The generics idiom, and erased generics, also depend on the observation
that in the full-generics semantics all instantiations of a signature
constructor have the \emph{same} members, only with different member
signatures\footnote{The Java compiler tries to compensate for the difference in members
signatures by using type cast operations that are never guaranteed
to fail.}.

Using our modeling of generics presented in earlier sections, in an
OO language where all object types have a common supertype (and without
type variables having bounds---see §\ref{sub:Bounded-Type-Variables-F-bounded-Generics}),
erasure is equivalent to instantiating a signature constructor with
one type then using the signature constructor name to refer to this
particular instantiation (rather than the ground signature name).

\global\long\def\ggtong{GndToNonGen}
The erasure of a ground signature name $ggnm=(nm,[\overline{ggnm}])$
\[
\ggtong(ggnm)=nm
\]
(That is, the erasure of the ground signature name $ggnm$ is simply
the name of the signature constructor used to construct $ggnm$\footnote{Erasure of $ggnm$ is a reminiscent reminder of the mathematical
\code{ceiling} function, $\left\lceil .\right\rceil $, which computes
the smallest integer greater than or equal to a floating point number.}).

If we abbreviate the $\ggtong$ function to \global\long\def\ggtong#1{\left\lceil #1\right\rceil }
$\ggtong{\cdot}$, then the erasure of a ground signature \textsf{$ggs=(ggnm,\mathsf{[\overline{ggnm}]},\mathsf{[\overline{ggfs}]},\mathsf{[\overline{ggms}]})$}\footnote{Erasure of a ground signature is like first converting the floating
number 3.2 to the floating number 3.0 (corresponding to instantiating
with \code{Object}), then referring to that new floating number using
(the symbol for) the integer 3 (by dropping type arguments of generic
signature names).} is\textsf{ 
\[
\ggtong{ggs}=(\ggtong{ggnm},[\overline{\ggtong{ggnm}}],[\overline{\ggtong{ggfs}}],[\overline{\ggtong{ggms}}])
\]
}For a signature constructor environment $sce=\{\ldots,(nm,sc),\ldots\}$,
its erasure is
\[
\ggtong{sce}=\{(nm,ngs)|(nm,sc)\in sce\wedge ngs=\ggtong{\gtog(sc,\overline{\mathtt{Object}})}\}
\]
(check Section~\ref{sub:Name-Substitution} for the definition of
the $\gtog$ function).

That is, the definition of erasure for signature constructor environments
converts a signature constructor in the environment to a ground signature
by instantiating it with a sequence of type arguments whose elements
all having the same value: namely, signature constructor name \code{Object}
(which we assume is the type at the top of the subtyping hierarchy).

Equivalently, without making a detour via ground signatures, we could
state that $ngs$ is the result of taking out the type variables component
of $sc$, replacing each plain (\emph{i.e.}, naked) type variable
by \code{Object}, and dropping the type arguments of all \emph{generic}
signature names (\emph{i.e.}, we only keep the outer signature constructor
name). Given that plain type variables are not allowed in the supersignatures
component of $sc$, this shows that the supersignatures component
of the result signature will use the\emph{ }same\emph{ }signature
constructors of the supersignatures component of $sc$ (This property
plays an important role in preserving the inheritance and subtyping
properties of the resulting erased-generic entities.)

For a generic object signature $gos=(ggnm,sce)$, its erasure, $egos$
is 
\[
egos=(\ggtong{ggnm},\ggtong{sce})
\]

\begin{thm}
If a $\GNOOP$ generic object signature $gos$ is well-formed, its
erasure, $egos$, is a well-formed non-generic $\NOOP$ object signature.\end{thm}
\begin{proof}
By induction on the structure of $gos$, and by noticing that $\ggtong{\cdot}$
produces $\NOOP$ signature name for given $\GNOOP$ ground signature
name, produces $\NOOP$ signatures for $\GNOOP$ signature constructors,
and produces $\NOOP$ signature environments for $\GNOOP$ signature
constructor environments.
\end{proof}
A domain of erased-generics objects constructed using these erased-generics
object signatures, based on $\NOOP$ domain equations (or $\GNOOP$
domain equations, for this matter), will produce a domain, $\edom O$,
isomorphic to the domain $\dom O$ of $\NOOP$ of non-generic objects.
It is this domain that is used to interpret generics in Java. The
erasure of generic information (\emph{i.e.}, its unavailability at
runtime) limits the operations a Java developer can perform on generic
types (at compile time) to ones that will be type sound even without
such information (dynamic dispatch is type safe, while \code{new}
and `cast' operations cannot be performed on plain type variables).

\subsubsection{\label{sub:Raw-Types}Java Raw Types}

`Raw types' refers to using non-structured\emph{ }plain signature
constructor names as generic signature names, simultaneously\emph{
}while using (structured) generic signature names. Raw types are an
artifact of Java erasure that are only motivated by attempting to
maintain ``migration compatibility''. Migration compatibility allows
Java developers to simultaneously use generic and non-generic code,
while motivating them to use generics. The Java type system gives
\emph{unmodified }non-generic code a generic interpretation. Raw types
 are best understood as generic signature names that are ``missing
type parametrization information''\footnote{Raw types are much like `unbounded wildcard types, or erased types,
with automatically-inserted type casts that are \emph{not} guaranteed
to succeed' (See Section~\ref{sub:Wildcards} for a brief discussion
of Java wildcards). Just like a Java program with downcasts, or ``stupid
casts'' is not guaranteed to be free of type errors (see~\cite{FJ/FGJ}),
also a Java program with raw types (which are used as generic types)
is not guaranteed (or cannot be decidably guaranteed) to be free of
type errors. A Java program with raw types that are always used as
\emph{non}-generic types, as is the case in pre-generics Java (\emph{i.e.},
Java 4.0-), can be decidably guaranteed to be free of type errors.
Generification helps the type system decide type safety. Generification
can somewhat be viewed as providing the type checker with the ``non-decidable''
portion of the analysis it does while type checking. Otherwise generification
would be an automatable process (Generification is believed to be
an undecidable problem).}. We relegate further analysis of raw types to future work.

\subsection{\label{sub:Bounded-Type-Variables-F-bounded-Generics}Bounded Type
Variables}

For more expressive and more precise typing, type variables in generic
OOP are usually provided with \emph{type bounds.} An upper type bound
on a type variable tells that all instantiations of the type variable
will be guaranteed to be subtypes of the provided bound. A lower type
bound of a type variable tells that the provided bound will be guaranteed
to be a subtype of all instantiations of the type variable. Bounds
on type variables of generic classes make use of the naturally-supported
subtyping relation in OOP. In this section, and thus in this thesis,
we only discusse upper type bounds of generic type variables, and
relegate a discussion of lower type bounds to future work.

An upper type bound of a type variable allows a class designer to
assume that certain members do exist in objects of any valid instantiation
of the type variable, and, because of nominality associating a behavioral
contract with type names, to also assume that all instantiations of
the type variable stick to the behavioral contract associated with
the name of the type bound.\footnote{For example, a class that provides (\emph{i.e.}, whose instances provide)
a generic sorting ``service'', needs to be sure that whatever elements
(objects) it orders do actually support ordering. The type of these
elements will be provided to the ``sorter'' generically, as a type
variable, but will be required, via an upperbound, to provide ``proof''
that its elements (of the type) do support ordering, \emph{e.g.},
by having a \code{compare()} or \code{leq()} (less-than-or-equal)
method.}

Formally, adding upperbounds (upper type bounds) to type variables,
amounts to a small addition to signature constructors, where the type
variables component is not a sequence of type variables (names from
set $\mathsf{X}$), but a sequence of pairs of type variables and
bounds, where bounds are members of the set $\mathsf{GN}$\footnote{Not the set $\mathsf{GNX}$. We do not allow type variables to bound
type variables, but only to appear as type arguments in the bounds.
Even though the consequences of lifting this restriction could be
interesting to investigate, we do not do so here.}.

We thus define
\[
\mathsf{UBX}=\mathsf{X}\times\mathsf{GN}
\]
 and we could now redefine signature constructors
\begin{eqnarray}
\mathsf{SC} & = & \mathsf{N}\times\mathsf{UBX}^{*}\times\mathsf{GN}^{*}\times(\mathsf{L}\times\mathsf{GNX})^{*}\times(\mathsf{L}\times\mathsf{GNX}^{*}\times\mathsf{GNX})^{*}\label{eq:sig-cons-ub}
\end{eqnarray}

Other than this small change, all other definitions we had in Section~\ref{sec:Generic-Signatures}
remain exactly the same\@. Type variable bounds, however, add few
more definitions and concepts.

\subsubsection{Valid Signature Constructor Instantiations and Valid Ground Signature
Names}

The main change done to our modeling of generics by having upper bounds
is that it introduces, on top of the notion of well-formed signature
constructor instantiations (which we have seen in Section~\ref{sec:Generic-Signatures}),
the notion of \emph{valid} signature constructor instantiations. Rather
than using generic object signatures that are only well-formed, valid
generic object signatures (which use valid signature constructor instantiations)
are now used to construct the domain of generic objects $\gdom O$
of $\GNOOP$ to model objects of generic OOP where type variables
have bounds. The addition of type variable bounds thus affects values
in the domain.

A valid signature constructor instantiation is a well-formed generic
object signature whose type arguments do satisfy the upper bound constraints
on the corresponding type variables they instantiate. A type argument
must be a subsign of the upper bound of the corresponding type variable.
This condition cannot be considered a well-formedness condition. As
we have seen in Section~\ref{sub:Well-formed-Generic-Signatures},
and Section~\ref{sub:Generic-Inheritance-and-Subsigning}, subsigning
is defined for well-formed signature constructor instantiations, by
\emph{implicitly} assuming they are instantiated with valid type arguments
(all instantiations were valid at the time, and thus there was no
need for an explicit validity check). Checking validity of an instantiation
depends on the subsigning relation being defined. If validity were
to be considered part of well-formedness then a circularity would
ensue between the definition of subsigning and well-formedness. Luckily,
this circularity can be resolved by introducing the notion of validity,
which is defined \emph{after} the subsigning relation is defined,
which, in turn, is defined after well-formedness is defined. Subsigning
(between two instantiations) is thus defined \emph{assuming} that
the instantiations are valid. A decision made regarding subsigning
does not later change if the either of the instantiations is discovered
to be invalid\footnote{This is the key observation behind not needing to redefine subsigning
for valid instantiations, and then redefine valid instantions and
so on ad infinitum.}.\footnote{This split between well-formedness (as a more syntactic condition)
and validity (as a more semantic condition), however, does introduce
the annoying, but inevitable, notion of ``well-formed yet invalid''
ground signature names. An example from Java is the ground signature
name \code{Enum<Object>}, which is a well-formed instantiation of
the generic \code{Enum} signature, but is invalid because \code{Object}
is not a subsign of the bound on \code{Enum}'s type variable (in
the case of this particular instantiation, the bound on the type variable
of class \code{Enum}, given its declaration in the Java class library,
would , a little bit confusingly, be \code{Enum<Object>} itself.
In Java, \code{Object} is not a subsignature of any other signature
but itself {[}its supersignatures component is empty{]}).}

Within the context of a signature constructor environment, thus, a
valid ground signature name is a well-formed ground signature name
with the additional condition that when the arguments are substituted
(via name substitution) in the upper bounds of the signature constructors
the required subsigning constraints will hold in the subsigning relation.

Formally, in the context of a signature constructor environment $sce$,
$ggn=(nm,[\overline{ggn}])$ if $\overline{ggn}$ are valid ground
signature names\footnote{Note that instantiations of zeroary signature constructors (only one
such ``instantiations'' is possible) are all valid.}, and if and only if when we define
\[
\overline{xub}=tvars(sce(nm))
\]
we have
\[
\overline{ggn}\subsign\{\overline{xub}\mapsto\overline{ggn}\}ub(\overline{xub})
\]
where $tvars$ is a function that projects a signature constructor
to its type variables component (now a member of $\mathsf{UBX}$),
$ub$ returns the second component (the upper bound) of a pair (subsigning
between pairs of ground signature constructor names is extended to
sequences of pairs in the obvious way {[}\emph{i.e.}, by conjunction{]}).
Note the use of $\overline{ggn}$ on both sides of the $\subsign$
relation, where it is used as the substitutant on the r.h.s of $\subsign$.

To construct a domain corresponding to these new instantiations, the
only change in the construction is to redefine domain of signatures
used to construct $\GNOOP$ to be the flat domain of valid generic
object signatures (which use valid ground signature names).

Given that we allow type variables to appear as arguments in their
bounds, substituting type arguments in their bounds then checking
if those same arguments subsign these instantiated bounds, allows
for a form of generics called \emph{F-bounded generics} (corresponding
to the notion of F-bounded polymorphism, of structural OOP). F-bounded
generics, arguably, provides more flexibility to generic nominal OOP
by allowing better typing of so-called ``binary methods''. Even
though we have many interesting ideas related to them, we relegate
further analysis of F-bounded generics, of so-called self-referential
F-bounded generics (like in class \code{Enum} of Java), and of the
generics-based solution to the problem of binary methods to future
work.

\subsubsection{Type Variable Bounds and Erasure}

Upper type bounds on type variables offer a chance to do a little
better regarding the information lost when interpreting generics using
erasure (See Section~\ref{sub:Erasure}). Due to the possible use,
inside a class, of the members of the bound, having a better erasure
precludes the use of ``guaranteed-to-not-fail'' type casts to make
such accesses (method calls, or field accesses).

In this case , because it is more precise (more informative) than
type \code{Object}, the erasure of the upper bound of a type variable
is itself used, rather than type \code{Object}, as the erasure of
the type variable when the type variable occurs naked inside a signature
constructor.

\subsection{\label{sub:Wildcards}Java Wildcards and Variance Annotations}

`Variance annotations' is a feature of OOP generics that helps resolve
the mismatch between subtyping polymorphism and generics, and to combine
the power of the two type-abstraction features.

Because of supporting mutation (which we do not model in $\NOOP$
nor $\GNOOP$), a type like \code{LinkedList<Integer>} and \code{LinkedList<Number>}
are unrelated by subtyping, even though type \code{Integer} is a
subtype of type \code{Number}. Variance annotations is a feature
that tries to resolve this mismatch. There are two styles for specifying
variance annotations: declaration-site variance annotations, and usage-site
variance annotations. Declaration-site variance annotations attach
``annotations'' to type variables of generic classes that limit
how type variables can be used in the class (eg, as type of arguments
to methods, or as return types of methods). Declaration-site variance
annotations are used in languages such as Scala and C\#. Usage-site
variance annotations attach annotations to type arguments of generic
class instantiations that limit how type variables used in the signature
of methods of instantiations of a generic class are viewed outside
(See~\cite{Igarashi02onvariance-based} for a more detailed discussion
of variance annotations). Java wildcards are the Java incarnation
of \emph{usage-site} variance annotations (See~\cite{Torgersen2004}). 

Given the dependence of variance annotations on the nominal subtyping
relation of OO languages, any accurate analysis of OO variance annotations
should depend on having a precise model of nominal subtyping. Thus,
having $\GNOOP$ should help better reasoning about wildcards.

In this section we do not dig into modeling wildcards, but having
$\GNOOP$ at hand enables us to consider more precisely what could
be involved if a full model of wildcards is developed.

Our most critical observation is that with wildcards (or variance
annotations, more generally), infinite chains of supertypes could
now occur in the subtyping relation of signature constructor instantiations.
Given the declaration of class \code{Enum} in Java, and with a typical
class declaration that has the heading \code{class C extends Enum<C>},
consider, for example, the chain of Java types \code{C}, which has
\code{Enum<C>} as a supertype (given the declaration of class \code{C}),
which, by subtyping rules for wildcard types, has supertype \code{Enum<? extends C>},
which in turn has \code{Enum<? extends Enum<C>\textcompwordmark{}>}
as a supertype, which has \code{Enum<? extends Enum<? extends C>\textcompwordmark{}>\textcompwordmark{}>}
as a supertype, which has \code{Enum<? extends Enum<? extends Enum<C>\textcompwordmark{}>\textcompwordmark{}>\textcompwordmark{}>}
as a supertype, ... and so on.

It is our opinion that having infinite supertype chains makes the
semantics of full wildcards problematic, pending a detailed development
of an elaborate model of wildcards. We do not see, however, a simpler
way to handle full wildcards.

\section{\label{sec:Polymorphic-Methods}Polymorphic Methods}

Polymorphic methods are usually associated with generics, but, despite
some similarity, polymorphic methods are not an artifact of generics
(generic classes) per se. Type variables of a generic class are shared
by all methods. Polymorphic methods provide method-specific type variables.
The declared type variables of a polymorphic method can be used in
the signature and the code of that polymorphic method only. Polymorphic
methods give more power and expressiveness to OO type systems.

Just like signature constructor type variables, method type variables
allow a method signature to be abstracted over a set of signature
names, allowing the same method to be used with different signature
name instantiations.

In this section we give few core formal definitions for formally modeling
polymorphic methods.

For signature constructors, we add an extra method type variables
component $\mathsf{Y}^{*}$ to sixth component of signature constructors
(\emph{i.e.}, their method signatures component), and we use a separate
set, $\mathsf{GNY}$, for types that can be used in method signatures.
\begin{eqnarray*}
\mathsf{SC} & = & \mathsf{N}\times\mathsf{X}^{*}\times\mathsf{GN}^{*}\times(\mathsf{L}\times\mathsf{GNX})^{*}\times(\mathsf{L}\times\mathsf{Y}^{*}\times\mathsf{GNY}^{*}\times\mathsf{GNY})^{*}
\end{eqnarray*}
\[
\mathsf{GN}=\mathsf{N}\times\mathsf{GNX}{}^{*}
\]
\[
\mathsf{GNX}=\mathsf{GN}+\mathsf{X}
\]
where for method signatures, we have
\[
\mathsf{GNM}=\mathsf{N}\times\mathsf{GNY}{}^{*}
\]
\[
\mathsf{GNY}=\mathsf{GNM}+\mathsf{X}+\mathsf{Y}
\]
(as expected, method type variables cannot be used except in the types
of the method that declares them, not in other parts of the signature
constructor, whether that is other methods, fields, or the supersignatures
of the signature constructor). Note that because of the $\mathsf{X}$
variant of $\mathsf{GNY}$ the type variables of a signature constructor
can be used inside a method signature.

We leave the rest of research work needed to model polymorphic methods
to future work. We expect it not to be too difficult since polymorphic
methods merely involve an additional level of schematic abstraction
and application.

\section{Concluding Remarks}

$\NOOP$ enables proper mathematical reasoning about mainstream OO
languages. Developing $\GNOOP$ based on $\NOOP$ has made `understanding
generics' an ``application'' demonstrating the utility of $\NOOP$
(\emph{i.e.}, of modeling nominal typing in particular). Reasoning
about generics is not possible using earlier models of OOP. Because
of them being structurally-typed and structurally-subtyped, earlier
models of OOP, in fact, do not allow us to evaluate any interesting
properties of nominal OO languages because the typing systems are
fundamentally incompatible.

Even without using $\GNOOP$ to reason about many generic OOP features,
one immediate value of developing it is it dispelling a misconception
that by supporting generics the type system of Java and similar languages
is `a complex hybrid of nominal and structural type systems' (See~\cite[p.254]{TAPL}).
Looking at $\GNOOP$, it is clear that all OO typing features of Java,
generic ones and ones unrelated to generics, are purely nominal. We
believe the misconception that Java has become `a hybrid' may have
been due to confusing the `structure' that generics does add to
signature/type names (which is indeed important in $\GNOOP$) and
the `structure' of objects and classes (\emph{i.e.}, the set of
object members and their signatures) which is the relevant structure
that structural type systems refer to and are interested in.

\section{Approach Extension}

The author is also currently considering extending the above approach
to generics and simplifying it. The extended approach is based on
what is tentatively called `nominal intervals' and `single-nested
generics'. A nominal interval is a type variable name with both upper
and lower bounds. In this approach, basically every type (including
wildcard types) written inside a class is ``captured'' (not certain
yet if that's the right term) into a synthetic type variable (a nominal
interval) of the class. For single-nested generics, because of the
capturing in synthetic type variables, a type argument (like any type
inside the class now) can always be a type variable (a nominal interval),
not a generic type. In this simplified approach, class names/type
names/signature constructor names will appear \emph{only} in bounds
of type variables (interval bounds).

For this approach, the equations for generic signatures will be along
lines of the following ($\mathsf{Y}$ is now the set of synthetic
and original type variable names):

\[
\mathsf{GN}=\mathsf{N}\times\mathsf{Y}{}^{*}
\]
\[
\mathsf{GNY}=\mathsf{GN}+\mathsf{Y}
\]
(Note single-nesting of generic types. No mutual dependency between
$\mathsf{GN}$ and $\mathsf{GNY}$ as that between $\mathsf{GN}$
and $\mathsf{GNX}$ above)

\[
\mathsf{YB}=\mathsf{Y}\times\mathsf{GN}\times\mathsf{GN}
\]
(Note nominal intervals. All type variables will have lower and upper
bounds, in addition to a name. Can be extended to $\mathsf{YB}=\mathsf{Y}\times\mathsf{GNY}\times\mathsf{GNY}$?)

\begin{eqnarray*}
\mathsf{SC} & = & \mathsf{N}\times\mathsf{YB}^{*}\times\mathsf{GN}^{*}\times(\mathsf{L}\times\mathsf{Y})^{*}\times(\mathsf{L}\times\mathsf{Y}^{*}\times\mathsf{Y})^{*}
\end{eqnarray*}
(Note: all types of fields and methods are now replaced by references
to type variables/nominal intervals).

We expect this approach to allow better modeling of Java wildcards
and Java erasure.

\bibliographystyle{plain}

\end{document}